\documentclass[11pt]{article}

\usepackage{graphicx,amssymb,amsmath}
\usepackage{amsfonts}
\usepackage{booktabs}
\usepackage{cite}
\usepackage{color}
\usepackage{enumerate}
\usepackage{float}
\usepackage{hyperref}
\usepackage{mathrsfs}
\usepackage{subfig}
\usepackage{tabularx}
\usepackage{authblk}

\usepackage[in]{fullpage}
\usepackage[top=0.8in, bottom=0.9in, left=0.9in, right=0.9in]{geometry}

\newcommand{\spm}{\mbox{$S\!P\!M$}}

\def\calP{\mathcal{P}}
\def\calF{\mathcal{F}}
\def\calS{\mathcal{S}}
\def\calU{\mathcal{U}}
\def\td{\tilde{d}}
\def\st{$s$-$t$}
\newtheorem{observation}{Observation}

\newtheorem{lemma}{Lemma}
\newtheorem{theorem}{Theorem}
\newenvironment{proof}{\noindent {\textbf{Proof:}}\rm}{\hfill $\Box$ \rm\bigskip}

\title{Shortest Paths Among Obstacles in the Plane Revisited\thanks{This research was supported in part by NSF under Grant CCF-2005323.}}

\author{
Haitao Wang
}
\affil{Department of Computer Science \\
Utah State University, Logan, UT 84322, USA
\\ {\tt haitao.wang@usu.edu}}

\begin{document}

\pagestyle{plain}
\pagenumbering{arabic}
\setcounter{page}{1}
\date{}

\thispagestyle{empty}
\maketitle

\vspace{-0.35in}
\begin{abstract}
Given a set of pairwise disjoint polygonal obstacles in the plane, finding an obstacle-avoiding  Euclidean shortest path between two points is a classical problem in computational geometry and has been studied extensively. The previous best algorithm was given by Hershberger and Suri [FOCS 1993, SIAM J. Comput. 1999] and the algorithm runs in $O(n\log n)$ time and $O(n\log n)$ space, where $n$ is the total number of vertices of all obstacles. The algorithm is time-optimal because $\Omega(n\log n)$ is a lower bound. It has been an open problem for over two decades whether the space can be reduced to $O(n)$. In this paper, we settle it by solving the problem in $O(n\log n)$ time and $O(n)$ space, which is optimal in both time and space; we achieve this by modifying the algorithm of Hershberger and Suri. Our new algorithm can build a shortest path map for a source point $s$ in $O(n\log n)$ time and $O(n)$ space, such that given any query point $t$, the length of a shortest path from $s$ to $t$ can be computed in $O(\log n)$ time and a shortest path can be produced in additional time linear in the number of edges of the path.
\end{abstract}


\section{Introduction}
\label{sec:intro}

Let $\calP$ be a set of pairwise disjoint polygonal obstacles with a total of $n$ vertices in the plane. The plane minus the interior of the obstacles is called the {\em free space}, denoted by $\calF$. Given two points $s$ and $t$ in $\calF$, we consider the problem of finding a Euclidean shortest path from $s$ to $t$ in $\calF$. This is a classical problem in computational geometry and has been studied extensively, e.g.,~\cite{ref:ChenCo13,ref:GhoshAn91,ref:HershbergerAn99,ref:MitchellA91,ref:MitchellSh96,ref:SharirOn86,ref:StorerSh94,ref:RohnertSh86,ref:GuibasOp89,ref:GuibasLi87,ref:HershbergerA91,ref:HershbergerCo94,ref:LeeEu84}.

Two main methods have been used to tackle the problem in the literature: the visibility graph method and the continuous Dijkstra method. The visibility graph method is to first construct the visibility graph of the vertices of $\calP$ along with $s$ and $t$, and then run Dijkstra's shortest path algorithm on the graph to find a shortest \st\ path. The best algorithms for constructing the visibility graph run in $O(n\log n+K)$ time~\cite{ref:GhoshAn91} or in $O(n+h\log^{1+\epsilon}h +K)$ time~\cite{ref:ChenA15} for any constant $\epsilon>0$, where $h$ is the number of obstacles of $\calP$ and $K$ is the number of edges of the visibility graph. Because $K=\Omega(n^2)$ in the worst case, the visibility graph method inherently takes quadratic time to solve the shortest path problem. To achieve a sub-quadratic time solution, Mitchell~\cite{ref:MitchellSh96} made a breakthrough and gave an $O(n^{3/2+\epsilon})$ time algorithm by using the continuous Dijkstra method. Also using the continuous Dijkstra method plus a novel conforming subdivision of the free space, Hershberger and Suri~\cite{ref:HershbergerAn99} presented an algorithm of $O(n\log n)$ time and $O(n\log n)$ space; the running time is optimal as $\Omega(n\log n)$ is a lower bound in the algebraic computation tree model. Hershberger and Suri~\cite{ref:HershbergerAn99} raised an open question whether the space of their algorithm can be reduced to $O(n)$.\footnote{An unrefereed report
\cite{ref:InkuluA10} announced an algorithm based on the
continuous Dijkstra approach with $O(n+h\log h\log n)$ time and $O(n)$ space.}

The continuous Dijkstra algorithms in both~\cite{ref:MitchellSh96} and~\cite{ref:HershbergerAn99} actually construct the {\em shortest path map}, denoted by $\spm(s)$, for a source point $s$. The map is of $O(n)$ size and can be used to answer shortest path queries. By building a point location data structure~\cite{ref:EdelsbrunnerOp86,ref:KirkpatrickOp83} on $\spm(s)$ in additional $O(n)$ time, given a query point $t$, the shortest path length from $s$ to $t$ can be computed in $O(\log n)$ time and a shortest \st\ path can be reported in time linear in the number of edges of the path.

The problem setting for $\calP$ is usually referred to as {\em polygonal domains} or {\em polygons with holes} in the literature. The shortest path problem in simple polygons is relatively easier~\cite{ref:GuibasOp89,ref:GuibasLi87,ref:HershbergerA91,ref:HershbergerCo94,ref:LeeEu84}. Guibas et al.~\cite{ref:GuibasLi87} gave an algorithm that can construct the shortest path map with respect to a source point in linear time. For two-point shortest path query problem where both $s$ and $t$ are query points, Guibas and Hershberger~\cite{ref:GuibasOp89,ref:HershbergerA91} built a data structure in linear time such that each query can be answered in $O(\log n)$ time. In contrast, the two-point query problem in polygonal domains is much more challenging: to achieve $O(\log n)$ time queries, the current best result uses $O(n^{11})$ space~\cite{ref:ChiangTw99}.

\subsection{Our result}

In this paper, by modifying the algorithm of Hershberger and Suri~\cite{ref:HershbergerAn99} (referred to as the {\em HS algorithm}), we show that the problem of finding a shortest path among obstacles in $\calP$ is solvable in $O(n\log n)$ time and $O(n)$ space, which is optimal in both time and space. This answers the longstanding open question of Hershberger and Suri~\cite{ref:HershbergerAn99}. Our algorithm actually constructs the shortest path map $\spm(s)$ for a source point $s$ in $O(n\log n)$ time and $O(n)$ space. We give an overview of our approach below.

The reason that the HS algorithm needs $O(n\log n)$ space is two-fold. First, it uses fully persistent binary trees (with the path-copying method) to represent wavefronts. Because there are $O(n)$ events in the wavefront propagation algorithm and each event costs $O(\log n)$ additional space on a persistent tree, the total space consumed in the algorithm is $O(n\log n)$. Second, in order to construct $\spm(s)$ after the propagation algorithm, some previous versions of the wavefronts are needed, which are maintained in those persistent trees.
We resolve these two issues in the following way.

We still use persistent trees to represent wavefronts. However, since there are $O(n)$ events in the propagation algorithm, we divide the algorithm into $O(\log n)$ phases such that each phase has no more than $n/\log n$ events. The total additional space for processing the events using persistent trees in each phase is $O(n)$. At the end of each phase, we ``reset'' the space of the algorithm by only storing a ``snapshot'' of the algorithm (and discarding all other used space) so that (1) the snapshot contains sufficient information for the subsequent algorithm to proceed as usual, and (2) the total space of the snapshot is $O(n)$. More specifically, the HS algorithm relies on a {\em conforming subdivision} of the free space $\calF$ to guide the wavefront propagation; our snapshot is composed of the wavefronts of a set of edges of the subdivision (intuitively, those are the edges in the forefront of the current wavefront).
In this way, we reset the space to $O(n)$ at the end of each phase. As such, the total space of the propagation algorithm is bounded by $O(n)$. This resolves the first issue.

For the second issue of constructing $\spm(s)$, the HS algorithm relies on some historical wavefronts produced during the propagation algorithm, plus some marked wavelet generators (a wavelet generator is either $s$ or a vertex of $\calP$) for each cell of the conforming subdivision; the total number of marked generators for all cells of the subdivision is $O(n)$. Due to the space-reset, our algorithm does not maintain historical wavefronts anymore, and thus we need to somehow restore these wavefronts. To this end, a key observation is that by marking a total of $O(n)$ additional wavelet generators it is possible to restore all historical wavefronts that are needed for constructing $\spm(s)$. In this way, $\spm(s)$ can be constructed in $O(n\log n)$ time and $O(n)$ space.


\paragraph{Outline.}
The rest of the paper is organized as follows. Section~\ref{sec:pre} defines notation and introduces some concepts. As our algorithm is a modification of the HS algorithm, we briefly review the HS algorithm in Section~\ref{sec:overview} and refer the reader to~\cite{ref:HershbergerAn99} for details. Our modified algorithm is described in Section~\ref{sec:newalgo}. Section~\ref{sec:con} concludes with remarks on some problem extensions and other possible applications of our technique.

\section{Preliminaries}
\label{sec:pre}

For any two points $p$ and $q$ in the free space $\calF$ of $\calP$, we use $\pi(p,q)$ to denote a shortest path from $p$ to $q$ in $\calF$. Note that $\pi(p,q)$ may not be unique, in which case $\pi(p,q)$ may refer to an arbitrary shortest path. Let $d(p,q)$ denote the length of $\pi(p,q)$; we call $d(p,q)$ the {\em geodesic distance} between $p$ and $q$. For two line segments $e$ and $f$ in the free space $\calF$, their {\em geodesic distance} is defined to be the minimum geodesic distance between any point on $e$ and any point on $f$, i.e., $\min_{p\in e, q\in f}d(p,q)$; by slightly abusing the notation, we use $d(e,f)$ to denote their geodesic distance.

For any two points $a$ and $b$ in the plane, denote by $\overline{ab}$ the line segment with $a$ and $b$ as endpoints; denote by $|\overline{ab}|$ the length of the segment.

Throughout the paper, let $s$ represent the source point. For any point $p$ in the free space $\calF$, the vertex prior to $p$ in $\pi(s,p)$ is called the {\em predecessor} of $p$ in $\pi(s,p)$.
The {\em shortest path map} $\spm(s)$ of $s$ is a decomposition of $\calF$ into maximal regions such that each region $R$ has an obstacle vertex that is the predecessor of all points in $R$; each edge of $R$ is either a fragment of an obstacle edge or a portion of a {\em bisector}, which is the locus of points $p$ with $d(s,u)+|\overline{pu}|=d(s,v)+|\overline{pv}|$ for two obstacle vertices $u$ and $v$ (we use $B(u,v)$ to denote their bisector). $B(u,v)$ is in general a hyperbola; a special case happens if one of $u$ and $v$ is the predecessor of the other, in which case $B(u,v)$ is on a straight line.

For any compact region $A$ of the plane, let $\partial A$ denote its boundary. We use $\partial \calP$ to denote the union of the boundaries of all obstacles of $\calP$.
We call the vertices of $\calP$ and $s$ the {\em obstacle vertices} and call the edges of $\calP$ the {\em obstacle edges}.

\paragraph{The conforming subdivision $\calS'$ of the free space.}
A novel contribution of the paper~\cite{ref:HershbergerAn99} is a conforming subdivision $\calS'$ of the free space $\calF$. We briefly review it here by following the notation in~\cite{ref:HershbergerAn99}.
Let $V$ denote the set of all obstacle vertices of $\calP$ plus the source point $s$.

The conforming subdivision $\calS'$ is built upon a conforming subdivision $\calS$ with respect to the points of $V$ (ignoring the obstacle edges). The subdivision $\calS$ is a quad-tree-style subdivision of the plane into $O(n)$ cells with the following properties. (1) Each cell is either a square or a square annulus  (i.e., an outer square with an inner square hole). (2) Each cell of $\calS$ contains at most one point of $V$ (only square cells can contain points of $V$). (3) Each edge of $\calS$ is axis-parallel. (3) Each edge $e$ of $\calS$ is {\em well-covered}, i.e., there exists a set of $O(1)$ cells whose union $\calU(e)$ contains $e$ with the following three properties: (a) the size of  $\calU(e)$ is $O(1)$; (b) for each edge $f$ on $\partial \calU(e)$, the Euclidean distance between $e$ and $f$ (i.e., the minimum $|\overline{pq}|$ among all points $p\in e$ and $q\in f$) is at least $2\cdot\max\{|e|,|f|\}$; (c) $\calU(e)$, which is called the {\em well-covering region} of $e$, contains at most one point of $V$.

The conforming subdivision $\calS'$ of the free space $\calF$ is built by inserting the obstacle edges into $\calS$.
More specifically, $\calS'$ is a subdivision of $\calF$ into $O(n)$ cells with the following properties. (1) Each cell of $\calS'$ is one of the connected components formed by intersecting $\calF$ with an axis-parallel rectangle (which is the union of a set of adjacent cells of $\calS$) or a square annulus of $\calS$. (2) Each cell of $\calS'$ contains at most one point of $V$. (3) $\calS'$ has two types of edges: {\em opaque edges}, which are fragments of obstacles edges, and {\em transparent edges}, which are introduced by the subdivision construction (these edges are in the free space $\calF$); each transparent edge is axis-parallel. (4) Each point of $V$ is incident to a transparent edge. (5) Each transparent edge $e$ of $\calS'$ is {\em well-covered}, i.e., there exists a set of $O(1)$ cells whose union $\calU(e)$ contains $e$ with the following three properties: (a) the size of $\calU(e)$ is $O(1)$; (b) for each edge $f$ on $\partial \calU(e)$, the geodesic distance $d(e,f)$ between $e$ and $f$ is at least $2\cdot\max\{|e|,|f|\}$; (c) $\calU(e)$, which is called the {\em well-covering region} of $e$, contains at most one point of $V$.

Both subdivisions $\calS$ and $\calS'$ can be constructed in $O(n\log n)$ time and $O(n)$ space~\cite{ref:HershbergerAn99}.



\section{A brief review of the HS algorithm}
\label{sec:overview}

In this section, we briefly review the HS algorithm. We refer the readers to~\cite{ref:HershbergerAn99} for details.

\begin{figure}[t]
\begin{minipage}[t]{\textwidth}
\begin{center}
\includegraphics[height=2.0in]{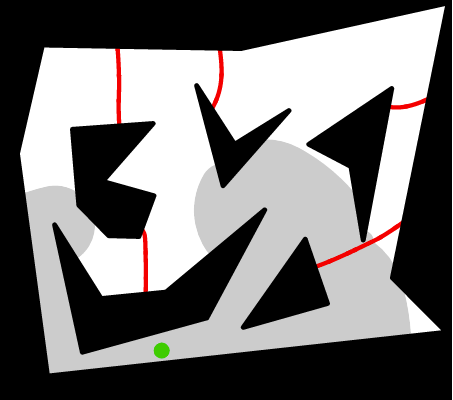}
\caption{\footnotesize Illustrating the wavefront. The black region are obstacles. The green point is $s$. The red curves are bisectors of $\spm(s)$. The gray region is the free space that has been covered by the wavefront. The boundary between the white region and the grey region is the wavefront. The figure is generated using the applet at~\cite{ref:HershbergerGe14}.}
\label{fig:wavefront}
\end{center}
\end{minipage}
\vspace{-0.15in}
\end{figure}


The algorithm starts generating a wavefront from the source point $s$ and uses the conforming subdivision $\calS'$ to guide the wavefront expansion. At any moment of the algorithm, the wavefront consists of all points of $\calF$ with the same geodesic distance from $s$ (e.g., see Fig.~\ref{fig:wavefront}). The wavefront is composed of a sequence of {\em wavelets}, each centered at an obstacle vertex that is already covered by the wavefront (the obstacle vertex is called the {\em generator} of the wavelet). To simulate the wavefront expansion, the algorithm produces the wavefront passing through each transparent edge of $\calS'$. As it is difficult to compute a true wavefront for each transparent edge $e$ of $\calS'$, a key idea of the HS algorithm is to compute two one-sided wavefronts (called {\em approximate wavefronts}) for $e$, each representing the wavefront coming from one side of $e$. Intuitively, an approximate wavefront from one side of $e$ is what the true wavefront would be if the wavefront were blocked off at $e$ by considering $e$ as an (open) opaque obstacle segment. To limit the interaction between approximate wavefronts from different sides of $e$, when an approximate wavefront propagates across $e$, an artificial wavelet is created at each endpoint of $e$. This is a mechanism to eliminate an approximate wavefront from one side of $e$ if it arrives at $e$ later than the approximate wavefront from the other side of $e$.

In the following, unless otherwise stated, a wavefront at a transparent edge $e$ of $\calS'$ refers to an approximate wavefront. We use $W(e)$ to denote a wavefront at $e$. As $e$ has two wavefronts, to make the description concise, depending on the context, $W(e)$ may refer to both wavefronts, i.e., the discussion on $W(e)$ applies to both wavefronts. For example, ``construct the wavefronts $W(e)$'' means construct the two wavefronts at $e$. Also, unless otherwise stated,
an edge of $\calS'$ refers to a transparent edge.

For each transparent edge $e$ of $\calS'$,
define $input(e)$ as the set of transparent edges on the boundary of $\calU(e)$ and define
$output(e)=input(e)\cup \{g\ |\ e\in input(g)\}$\footnote{The set $input(e)$ is included in $output(e)$ because the algorithm relies on the property that edges of $output(e)$ form a cycle enclosing $e$.}. As the size of $\calU(e')$ for each transparent edge $e'$ of $\calS'$ is $O(1)$, both $|input(e)|$ and $|output(e)|$ are $O(1)$.
The wavefronts at $e$ are computed from the wavefronts at all edges of
$input(e)$; this guarantees the correctness because $e$ is in $\calU(e)$ (and thus
the shortest path $\pi(s,p)$ passes through some edge $f\in input(e)$ for any point $p\in e$).
After the wavefronts at $e$ are computed, they will pass
to the edges of $output(e)$. In the meanwhile, the geodesic distances
from $s$ to both endpoints of $e$ will be computed. By the property (4) of $\calS'$, each
obstacle vertex is incident to a transparent edge of $\calS'$. Hence,
after the algorithm is finished, geodesic distances from $s$ to all
obstacle vertices will be available. The process of passing the
wavefronts $W(e)$ at $e$ to all edges $g\in output(e)$ is called the
{\em wavefront propagation procedure} (see Section 5 in~\cite{ref:HershbergerAn99}), which will compute the wavefront $W(e,g)$,
the portion of $W(e)$ that passes to $g$ through the well-covering region $\calU(g)$ of $g$ if $e\in input(g)$ and through $\calU(e)$ otherwise;
whenever the procedure is invoked on $e$, we say that $e$ is {\em
processed}. The wavefronts $W(e)$ at $e$ are constructed by
merging the wavefronts $W(f,e)$ for all edges $f\in input(e)$; this process is called the
{\em wavefront merging procedure}  (see Lemma 4.6 in~\cite{ref:HershbergerAn99}).

The transparent edges of $\calS'$ are processed in a rough time order.
The wavefronts $W(e)$ of each edge $e$ are constructed at
the time $\td(s,e)+|e|$, where $\td(s,e)$ is the minimum
geodesic distance from $s$ to the two endpoints of $e$. At the time $\td(s,e)+|e|$, it is guaranteed that all edges $f\in input(e)$ whose wavefronts $W(f)$ contribute a wavelet to $W(e)$ have already been processed (see Lemma~4.2~\cite{ref:HershbergerAn99}); this is due to the property 5(b) of $\calS'$ that $d(e,f)\geq 2\cdot \max\{|e|,|f|\}$ (because $f$ is on $\partial\calU(e)$).

Define $covertime(e)=\td(s,e)+|e|$. The value $covertime(e)$ will be computed during the algorithm. Initially, for each edge $e$ whose well-covering region contains $s$, $covertime(e)$ is computed directly. The algorithm maintains a timer $\tau$ and processes the transparent edges $e$ of $\calS'$ following the order of $covertime(e)$. The main loop of the algorithm works as follows. Recall that ``processing an edge $e$'' means invoking the wavefront propagation procedure on $W(e)$.
As long as $\calS'$ has an unprocessed edge, we do the following. First, among all unprocessed edges, choose the one $e$ with minimum $covertime(e)$ and set $\tau=covertime(e)$. Second, call the
wavefront merging procedure to construct the wavefronts $W(e)$ from $W(f,e)$ for all edges $f\in input(e)$ satisfying $covertime(f)<covertime(e)$; compute $d(s,v)$  from $W(e)$ for each endpoint $v$ of $e$. Third, process $e$, i.e., call the wavefront propagation procedure on $W(e)$ to compute $W(e,g)$ for all edges $g\in output(e)$; in particular, compute the time $\tau_g$ when the wavefronts $W(e)$ first encounter an endpoint of $g$ and set $covertime(g)= \min\{covertime(g),\tau_g+|g|\}$.



Each wavefront $W(e)$ is represented by a persistent balanced binary trees by path copying~\cite{ref:DriscollMa89} so that we can operate on an old version of $W(e)$ even after it is used to construct wavefronts $W(e,g)$ for some $g\in output(e)$. Due to the path-copying of the persistent data structure, each operation on $W(e)$ will cost $O(\log n)$ additional space. Note that in the wavefront merging procedure,
because we do not need the old versions of $W(f,e)$ after $W(e)$ is constructed, it is not necessary to use the path-copying method after each operation on $W(f,e)$ in the procedure. Hence, the total additional space in the wavefront merging procedure in the entire algorithm is $O(n)$. But $O(n\log n)$ additional space is needed in the wavefront propagation procedure because after it is used to compute $W(e,g)$ for some $g\in output(e)$, the old version of $W(e)$ needs to be kept for computing $W(e,g')$ for other $g'\in output(e)$.
There are $O(n)$ {\em bisector events} in the wavefront propagation procedure in the entire algorithm. In particular, if two generators $u$ and $v$ are not adjacent in $W(e)$ but become adjacent during the propagation from $e$ to an edge $g\in output(e)$, then there is a bisector event involving $u$ and $v$. The intersection between a bisector and an obstacle edge will also cause a bisector event. There are a total of $O(n)$ operations on persistent trees during the wavefront propagation procedure (we call them {\em wavefront propagation operations}) and each such operation requires $O(\log n)$ additional space in the persistent trees; thus the total space needed in the wavefront propagation procedure is $O(n\log n)$.

The algorithm halts once all transparent edges of $\calS'$ are processed. Then,
geodesic distances from $s$ to all obstacle vertices are computed. This is actually the first main step of the HS algorithm, which we refer to as the {\em wavefront expansion step} (we could call it the wavefront propagation step, but to distinguish it from the wavefront propagation procedure, we use ``expansion'' instead). The second main step of the HS algorithm is to construct the shortest path map $\spm(s)$, which we call the {\em SPM-construction} step.

During the above wavefront expansion step, some generators are marked so that if a generator $v$ is involved in a true bisector event of $\spm(s)$ in a cell $c$ of $\calS'$, then $v$ is guaranteed to be marked for $c$. Also, after the wavefront expansion step, thanks to the persistent binary trees, the wavefronts $W(e)$ for all transparent edges $e$ of $\calS'$ are still available. To construct $\spm(s)$, its vertices in each cell $c$ of $\calS'$ are computed. To this end, the marked generators for $c$ and the wavefronts for the transparent edges on the boundary of $c$ are utilized, as follows.

First, $c$ is partitioned into {\em active} and {\em inactive} regions so that no vertices of $\spm(s)$ are in inactive regions. Because only marked generators possibly contribute to true bisector events in $c$, the portion of a bisector defined by a marked generator and an unmarked generator in $c$ belongs to an edge of $\spm(s)$. All such bisectors, which must be disjoint in $c$, partition $c$ into regions each of which is claimed either by only marked generators or by only unmarked generators; the former regions are active while the latter are inactive.

Consider an active region $R$, whose boundary consists of $O(1)$ segments, each of which is a transparent edge fragment, an obstacle edge fragment, or a bisector portion in $\spm(s)$. For each transparent edge fragment $e$ of $\partial R$, we use $W(e)$ to partition $R$ into sub-regions, each with a unique predecessor in $W(e)$. This can be done in $O(|W(e)|\log |W(e)|)$ time by propagating the wavefront $W(e)$ into $R$ using the wavefront propagation procedure (and ignoring the wavefronts of other transparent edges of $\partial R$). Let $S(e)$ denote the partition of $R$.

\paragraph{Remark.} Since the wavefront $W(e)$ is not useful anymore after the above step, here in the wavefront propagation algorithm for computing $S(e)$ we do not need to use the path-copying method for each operation on the tree representing $W(e)$. Therefore, constructing the partition $S(e)$ takes only linear additional space. Consequently, the total additional space for constructing $\spm(s)$ is $O(n)$. Although this does not matter for the HS algorithm, it helps in our new algorithm.
\bigskip

The partitions $S(e)$ for all transparent edges $e\in \partial R$ are merged in a similar way as the merge step in the standard divide-and-conquer Voronoi diagram algorithm~\cite{ref:ShamosCl75}; during the merge process, the vertices of $\spm(s)$ in $R$ are computed (see Lemma~4.12 in~\cite{ref:HershbergerAn99}).

Applying the above for all active regions of $c$ computes all vertices of $\spm(s)$ in $c$.

After the vertices of $\spm(s)$ in all cells $c\in \calS'$ are computed, which in total takes $O(n\log n)$ time and $O(n)$ space as remarked above, $\spm(s)$ can be constructed in $O(n\log n)$ additional time and $O(n)$ space by first computing all edges of $\spm(s)$ and then assembling the edges using a standard plane sweep algorithm (see Lemma~4.13 in~\cite{ref:HershbergerAn99}).

In summary, the SPM-construction step of the HS algorithm runs in $O(n\log n)$ time and $O(n)$ space.

\section{Our new algorithm}
\label{sec:newalgo}

In this section, we present our modified algorithm, which uses $O(n\log n)$ time and $O(n)$ space. It turns out that we have to modify both main steps of the HS algorithm.
We first give an overview of our approach in Section~\ref{subsec:overview}.  Then we describe our modification to the wavefront expansion step in Section~\ref{subsec:expansion} while the proof of a key lemma is given in Section~\ref{subsec:lemma}. Section~\ref{sublem:spm} discusses the SPM-construction step.

\subsection{Overview}
\label{subsec:overview}

Our modification to the wavefront expansion step is mainly on the wavefront propagation procedure. As there are a total of $O(n)$ wavefront propagation operations, we divide the algorithm into $O(\log n)$ phases such that each phase has no more than $n/\log n$ operations. We still use persistent binary trees to represent wavefronts $W(e)$. To reduce the space, at the end of each phase, we ``reset'' the space of the algorithm by storing a ``snapshot'' of the algorithm and discarding the rest of the used space. The snapshot consists of wavefronts $W(e)$ of a set of transparent edges $e$ of $\calS'$. We show that these wavefronts are sufficient for the subsequent algorithm to proceed as usual. We prove that the total space of the snapshot is $O(n)$; this is a main challenge of our approach and Section~\ref{subsec:lemma} is devoted to this. Since each phase has $n/\log n$ wavefront propagation operations, the total additional space in each phase is $O(n)$. Due to the space reset, the total space used in the algorithm is $O(n)$.

For the SPM-construction step, it considers each cell of $\calS'$ individually. For each cell $c$, the algorithm has two sub-steps. First, compute the active regions of $c$. Second, for each active region $R$, compute the vertices of $\spm(s)$ in $R$. For both sub-steps, the original HS algorithm utilizes the wavefronts of the transparent edges on the boundary of $c$. Due to the space reset, the wavefronts are not available anymore in our new algorithm. We use the following strategy to resolve the issue. First, to compute the active regions in $c$, we need to know the bisectors defined by an unmarked generator $u$ and a marked generator $v$. We observe that $u$ is a generator adjacent to $v$ in the wavefronts along the boundary of $c$. Based on this observation, in our new algorithm we will mark the neighbors of the generators originally marked in the HS algorithm and we call them {\em newly-marked} generators (the generators marked in the original HS algorithm are called {\em originally-marked} generators). We show that the newly-marked generators and the originally-marked generators are sufficient for computing all active regions of each cell $c$, and the total number of all marked generators is still $O(n)$. Second, to compute the vertices of $\spm(s)$ in each active region $R$ of $c$, we need to restore the wavefronts of the transparent edges on the boundary of $R$. To this end, we observe that these wavefronts are exactly determined by the originally-marked generators. Consequently, the same algorithm as before can be applied to construct $\spm(s)$ in $O(n\log n)$ time, but the space is only $O(n)$.


\subsection{The wavefront expansion step}
\label{subsec:expansion}

In the wavefront expansion step, our main modification is on the wavefront propagation procedure, which is to compute the wavefronts $W(e,g)$ for all edges $g\in output(e)$ using the wavefront $W(e)$.

We now maintain a counter $count$ in the wavefront expansion step to record the number of wavefront propagation operations that have been executed so far since the last space reset; $count=0$ initially. Consider a wavefront propagation procedure on a wavefront $W(e)$ of a transparent edge $e$. The algorithm will compute $W(e,g)$ for each edge $g\in output(e)$, by propagating $W(e)$ through the well-covering region $\calU(g)$ of $g$ if $e\in input(g)$ and through $\calU(e)$ otherwise. We apply the same algorithm as before. For each wavefront propagation operation, we first do the same as before. Then, we increment $count$ by one. If $count< n/\log n$, we proceed as before (i.e., process the next wavefront operation). Otherwise, we have reached the end of the current phase and start a new phase. To do so, we first reset $count=0$ and then reset the space by constructing and storing a snapshot of the algorithm, as follows.

\begin{enumerate}
\item
Let $g$ refer to the edge of $output(e)$ whose $W(e,g)$ is currently being computed in the algorithm. We store the tree that is currently being used to compute $W(e,g)$ right after the above wavefront propagation operation. To do so, we can make a new tree by copying the newest version of the current persistent
tree the algorithm is operating on, and thus the size of the tree is $O(n)$. We will use this tree to ``resume'' computing $W(e,g)$ in the subsequent algorithm.

\item
For each $g'\in output(e)\setminus\{g\}$ whose $W(e,g')$
has been computed, we store the tree for $W(e,g')$. We will use the tree to
compute the wavefronts $W(g')$ of $g'$ in the subsequent algorithm.

\item
We store the tree for the wavefront $W(e)$. Note that the tree may have many versions due to performing the wavefront propagation operations and we only keep its original version for $W(e)$. Hence, the size of the tree is $O(|W(e)|)$. This tree
will be used in the subsequent algorithm to compute $W(e,g')$ for those edges $g'\in output(e)\setminus\{g\}$ whose
$W(e,g')$ have not been computed.

\item
We check every transparent edge $e'$ of $\calS'$ with $e'\neq e$. If $e'$ has been
processed (i.e., the wavefront propagation procedure has been called
on $W(e')$) and there is an edge $g'\in output(e')$ that has {\em not} been processed, we know that $W(e',g')$ has been computed and is available; we store the tree for $W(e',g')$. We will use the tree to compute the wavefronts $W(g')$ of $g'$ in the subsequent algorithm.
\end{enumerate}

We refer to the above four steps as the {\em snapshot construction algorithm}. We refer to the wavefronts stored in the algorithm as the {\em snapshot}; intuitively, the snapshot contains all wavelets in the forefront of the current global wavefront.

The following lemma shows that if we discard all persistent trees currently used in the algorithm and instead store the snapshot, then the algorithm will proceed without any issues.

\begin{lemma}
The snapshot stores sufficient information for the subsequent algorithm to proceed as usual.
\end{lemma}
\begin{proof}
Let $\xi$ denote the moment of the algorithm when the snapshot construction algorithm starts.
Let $\xi'$ be any moment of the algorithm after the snapshot is constructed. In the following, we show that if the algorithm needs any information that was computed before $\xi$ in order to perform certain operation at $\xi'$, then that information is guaranteed to be stored at the snapshot. This will prove the lemma.

At $\xi$, the transparent edges of $\calS'$ excluding $e$ can be
classified into two categories: those that have already been processed
at $\xi$ and those that have not been processed. Let $E_1$ and $E_2$
denote the sets of the edges in these two categories, respectively. The edge $e$ is special in the sense that it is currently being processed at $\xi$.

At $\xi'$, the algorithm may be either in the wavefront merging procedure or in the wavefront propagation procedure. We discuss the two cases separately.

\paragraph{The wavefront merging procedure case.}
Suppose the algorithm is in the wavefront merging procedure at $\xi'$,
which is to compute $W(e')$ for some edge $e'$ by merging all
wavefronts $W(f',e')$ for $f'\in input(e')$. Because at $\xi'$ we need
some information that was computed before $\xi$, that information must
be the wavefront $W(f',e')$ for some edge $f'\in input(e')$.
Depending on whether $f'=e$, there are two
subcases.

\begin{itemize}
\item
If $f'=e$, then since $W(e,e')$ was computed before $\xi$, $W(e,e')$ is stored in the snapshot by step~(2) of the snapshot construction algorithm.

\item
If $f'\neq e$, then since $W(f',e')$ was computed before $\xi$, the
edge $f'$ must have been processed at $\xi$.
Because the algorithm is computing $W(e')$ at $\xi'$, $e'$ has not been
processed at $\xi$. Therefore, $W(f',e')$ is stored in the snapshot by
step~(4) of the snapshot construction algorithm.
\end{itemize}

Hence, in either subcase, the information needed by the algorithm at $\xi'$ is stored at the snapshot.

\paragraph{The wavefront propagation procedure case.}
Suppose the algorithm is in the wavefront propagation procedure at
$\xi'$, which is to process a transparent edge $e'$, i.e., compute $W(e',g')$ for
some $g'\in output(e')$. Because at $\xi'$ the algorithm needs some information that
was computed before $\xi$ to compute $W(e',g')$, $W(e')$ must have been computed before $\xi$ (since $W(e',g')$ relies on $W(e')$).

We claim that $e'$ must be $e$. Indeed, if $e'\in E_1$, then $e'$ has
been processed before $\xi$ and thus the wavefront propagation
procedure cannot happen to $e'$ at $\xi'$, a contradiction. If $e'\in
E_2$, then $e'$ has not been processed at $\xi$. According to our
algorithm, for any transparent edge $e''$, $W(e'')$ is computed during
the wavefront merging procedure for $e''$, which is immediately
followed by the wavefront propagation procedure to process $e''$.
Since at $\xi$ the algorithm is in the wavefront propagation procedure
to process $e$, the wavefront merging procedure for $e'$ must have not
been invoked at $\xi$, and thus $W(e')$ must have not been computed at
$\xi$. This contradicts with the fact that $W(e')$ has been computed
at $\xi$. Therefore, $e'$ must be $e$.

Depending on whether $g'$ is $g$, there are two subcases.
\begin{itemize}
\item
If $g'=g$, then the tree at the moment $\xi$ during the propagation for computing $W(e,g)$ from $W(e)$ is stored in the snapshot by the step (1) of the snapshot construction algorithm, and thus we can use the tree to ``resume'' computing $W(e,g)$ at $\xi'$.

\item
If $g'\neq g$, then in order to compute $W(e,g')$ at $\xi'$, we need the wavefront $W(e)$, which is stored in the snapshot by step~(3) of the snapshot construction algorithm.
\end{itemize}

Hence, in either subcase, the information needed by the algorithm at $\xi'$ is stored at the snapshot.

The lemma thus follows.
\end{proof}

A challenging problem is to bound the space of the snapshot, which is established in the following lemma. As the proof is lengthy and technical, we devote the next subsection to it.

\begin{lemma}\label{lem:space}
The total space of the snapshot is $O(n)$.
\end{lemma}

Since each phase has no more than $n/\log n$ wavefront propagation operations, the total extra space introduced by the persistent trees in each phase is $O(n)$. Due to the space-reset and Lemma~\ref{lem:space}, the total space of the algorithm is $O(n)$.

For the running time of our new algorithm, comparing with the original HS algorithm we spend extra time on constructing the snapshot at the end of each phase.
In light of Lemma~\ref{lem:space}, each call of the snapshot construction algorithm takes $O(n)$ time. As there are $O(\log n)$ phases, the total time on constructing the snapshots in the entire algorithm is $O(n\log n)$. Hence, the running time of our new algorithm is still bounded by $O(n\log n)$.

\subsection{Proof of Lemma~\ref{lem:space}}
\label{subsec:lemma}

In this subsection, we prove Lemma~\ref{lem:space}, i.e., prove that the space introduced by the four steps of the snapshot
construction algorithm is $O(n)$.

For the first step, as each wavefront has $O(n)$ generators, the size
of one tree is $O(n)$. For the second step, since $|output(e)|=O(1)$
and the size of each $W(e,g')$ is $O(n)$, the total space is $O(n)$.
For the third step, since $|W(e)|=O(n)$, the space is $O(n)$.
Below we focus on the fourth step.

Let $\Pi$ denote the collection of pairs $(e,g)$
whose wavefront $W(e,g)$ is stored in the fourth step of the
snapshot construction algorithm. Our goal is to show that $\sum_{(e,g)\in
\Pi}|W(e,g)|=O(n)$.

For an edge $e$, $\Pi$ may have multiple pairs with $e$ as
the first element and the second elements of all these pairs are in
$output(e)$; among all those pairs, we only keep the pair $(e,g)$
such that $|W(e,g)|$ is the largest in $\Pi$ and remove all other pairs from $\Pi$. Since
$|output(e)|=O(1)$, it suffices to show that the total sum of
$|W(e,g)|$ for all pairs $(e,g)$ in the new $\Pi$ is $O(n)$. Now in
the new $\Pi$, no two pairs have the same first element.
However, for an edge $g$, it is possible that there are multiple
pairs in $\Pi$ whose second elements are all $g$ and their first
elements are all from $input(g)$; among all those pairs, we only keep the pair $(e,g)$
such that $|W(e,g)|$ is the largest in $\Pi$ and remove all other pairs from $\Pi$. Since
$|input(g)|=O(1)$, it suffices to show that the total sum of
$|W(e,g)|$ for all pairs $(e,g)$ in the new $\Pi$ is $O(n)$. Now in
the new $\Pi$, no two pairs have the same first element and no two
pairs have the same second element. Hence, the size of $\Pi$ is
$O(n)$.

For each pair $(e,g)\in \Pi$, recall that $W(e,g)$ may have up to two
artificial generators (i.e., the endpoints of $e$). Since
$|\Pi|=O(n)$, there are a total of $O(n)$ artificial generators in
$W(e,g)$ for all pairs $(e,g)\in \Pi$. In the following discussion we
ignore these artificial generators and by slightly abusing the notation we use $W(e,g)$ to refer to the
remaining wavefront without the artificial generators.
Hence, each generator of $W(e,g)$ is an obstacle vertex.
It suffices to prove $\sum_{(e,g)\in \Pi}|W(e,g)|=O(n)$.

For any three adjacent generators $u$, $v$, and $w$ that are obstacle vertices in a wavefront, we call $(u,v,w)$ an {\em adjacent-generator-triple}. Two adjacent-generator-triples $(u_1,v_1,w_1)$ and $(u_2,v_2,w_2)$ are {\em distinct} if $u_1\neq u_2$, or $v_1\neq v_2$, or $w_1\neq w_2$.
We have the following observation.

\begin{observation}\label{obser:triple}
The total number of distinct adjacent-generator-triples in all wavefronts involved in the entire wavefront expansion step of the HS algorithm is $O(n)$.
\end{observation}
\begin{proof}
Initially, the algorithm starts a wavefront from $s$ with only one generator $s$. Whenever a distinct adjacent-generator-triple is produced during the algorithm, a bisector event must happen. As there are $O(n)$ bisector events~\cite{ref:HershbergerAn99}, the observation follows.
\end{proof}

Lemma~\ref{lem:space} is almost an immediate consequence of the following lemma.
\begin{lemma}\label{lem:gentriple}
Any adjacent-generator-triple $(u,v,w)$ can appear in the wavefront $W(e,g)$ for at most $O(1)$ pairs $(e,g)$ of $\Pi$.
\end{lemma}

Before proving Lemma~\ref{lem:gentriple}, we prove Lemma~\ref{lem:space} with the help of Lemma~\ref{lem:gentriple}. Indeed, there are $O(n)$ distinct adjacent-generator-triples in the entire algorithm by Observation~\ref{obser:triple}. Now that each such triple can only appear in a constant number of wavefronts $W(e,g)$ for all $(e,g)\in\Pi$, the total number of generators in all wavefronts $W(e,g)$ for all $(e,g)\in\Pi$ is bounded by $O(n)$. This leads to Lemma~\ref{lem:space}.

\paragraph{Proving Lemma~\ref{lem:gentriple}.}
We prove Lemma~\ref{lem:gentriple} in the rest of this subsection.
Assume to the contrary that an adjacent-generator-triple $(u,v,w)$ appears in the wavefront $W(e,g)$ for more than $O(1)$ pairs $(e,g)$ of $\Pi$ \footnote{By ``more than $O(1)$ pairs'', we intend to say ``more than $c$ pairs for a constant $c$ to be fixed later''. To simplify the discussion, we use ``more than $O(1)$ pairs'' instead with the understanding that such a constant $c$ can be fixed.}; let $\Pi'$ denote the set of all such pairs.

Consider a pair $(e,g)\in \Pi'$. Recall that $W(e,g)$ is obtained by propagating $W(e)$ from $e$ to $g$ through the well-covering region $\calU(g)$ if $e\in input(g)$ and through $\calU(e)$ otherwise. If one of $u$, $v$, and $w$ is in $\calU(g)\cup \calU(e)$, then we call $(e,g)$ a {\em special pair}. 
The properties of the subdivision $\calS'$ guarantee that each obstacle vertex is in $\calU(f)$ for at most $O(1)$ edges $f$ of $\calS'$~\cite{ref:HershbergerAn99}, and thus $\Pi'$ has at most $O(1)$ special pairs. We remove all special pairs from $\Pi'$, after which $\Pi'$ still has more than $O(1)$ pairs and for each pair $(e,g)\in \Pi'$ the three generators $u$, $v$, and $w$ are all outside $\calU(g)\cup \calU(e)$.
Therefore, $u$, $v$, and $w$ are also generators in $W(e)$ in this order (although they may not be adjacent in $W(e)$). Because $(u,v,w)$ is an adjacent-generator-triple in $W(e,g)$, $B(u,v)$ (resp., $B(v,w)$)  intersects $g$. Also, note that $u$, $v$, and $w$ are on the same side of the supporting line of $e$; since they are generators of both $W(e)$ and $W(e,g)$, the bisector $B(u,v)$ (resp., $B(v,w)$) intersects $e$ as well (although the intersection may not appear in $W(e)$, i.e., the intersection may be covered by a different generator in $W(e)$).
Further, $v$ is closer to the intersection of $B(u,v)$ (resp., $B(v,w)$) with $e$ that with $g$.
Let $(e_1,g_1)$ be the pair in $\Pi'$ such that the intersection of $B(u,v)$ with its first element is closest to $v$ among the intersections of $B(u,v)$ with the first elements of all pairs of $\Pi'$.

By the property 5(a) of $\calS'$, the size of $\calU(g_1)\cup \calU(e_1)$ is $O(1)$.
Because $|\Pi'|$ is not $O(1)$, $\Pi'$ must have a pair, denoted by $(e_2,g_2)$, such that $e_2$ is outside $\calU(g_1)\cup \calU(e_1)$. Recall that $W(e_1,g_1)$ is obtained by propagating $W(e_1)$ from $e_1$ to $g_1$ inside $\calU(g_1)$ or $\calU(e_1)$. Hence, if we move along the bisector $B(u,v)$ from its intersection with $e_1$ to its intersection with $g_1$, all encountered edges of $\calS'$ are in $\calU(g_1)\cup \calU(e_1)$. Therefore, by the definitions of $e_1$ and $e_2$, $B(u,v)$ intersects $e_1$, $g_1$, $e_2$, and $g_2$ in this order following their distances from $v$ (e.g., see Fig.~\ref{fig:bisectors}).

By definition (i.e., the fourth step of our snapshot construction algorithm), $e_2$ has been processed but $g_1$ has not. According to the wavefront expansion algorithm, $covertime(e_2)\leq covertime(g_1)$. In the following we will obtain $covertime(e_2)>covertime(g_1)$, which leads to a contradiction.

Let $b$ be the intersection between the bisector $B(u,v)$ and $g_2$ (e.g., see Fig.~\ref{fig:bisectors}). Since $u$ and $v$ are two adjacent generators in $W(e_2,g_2)$, $\overline{vb}$ is in the free space $\calF$.  As $W(e_2,g_2)$ is obtained from $W(e_2)$ and $v$ is also a generator in $W(e_2)$, $\overline{vb}$ must intersect $e_2$, say, at a point $q$ (e.g., see Fig.~\ref{fig:bisectors}).
Further, we have the following observation.

\begin{observation}\label{obser:intersect}
$\overline{vb}$ intersects $g_1$ (e.g., see Fig.~\ref{fig:bisectors}).
\end{observation}
\begin{proof}
Since $(u,v,w)$ is an adjacent-generator-triple in both wavefronts $W(e_1,g_1)$ and $W(e_2,g_2)$, each of the bisectors $B(u,v)$ and $B(v,w)$ intersects both $g_1$ and $g_2$. Let $z_1$ and $z_2$ be the intersections of $g_1$ with $B(u,v)$ and $B(v,w)$, respectively. Let $b'$ be the intersection of $g_2$ and $B(v,w)$. Then, $\overline{z_1z_2}$, which is a subsegment of $g_1$, is {\em claimed} by $v$ in $W(e_1,g_1)$ (i.e., among all wavelets of $W(e_1,g_1)$, the wavelet at $v$ reaches $p$ the earliest, for all points $p\in \overline{z_1z_2}$). Similarly, $\overline{bb'}$, which is a subsegment of $g_2$, is claimed by $v$ in $W(e_2,g_2)$. Recall that $v$ is closer to $z_1$ than to $b$. Hence, for any point $p\in \overline{bb'}$, $\overline{pv}$ must intersect $\overline{z_1z_2}$. This leads to the observation, for $\overline{z_1z_2}\subseteq g_2$ and $b\in \overline{bb'}$.
\end{proof}


In light of the above observation, let $z$ be an intersection point of $\overline{vb}$ and $g_1$ (e.g., see Fig.~\ref{fig:bisectors}). Using the properties of the well-covering regions of $\calS'$, we have the following observation.

\begin{observation}\label{obser:length}
$|\overline{qz}|\geq 2\cdot |g_1|$.
\end{observation}
\begin{proof}
Since $q\in e_2$ is outside the well-covering region $\calU(g_1)$ of $g_1$, $z\in g_1$, and $\overline{qz}$ is in the free space $\calF$, $\overline{qz}$ must cross the boundary of $\calU(g_1)$, say, at a point $p$. By the property 5(b) of the conforming subdivision $\calS'$, $|\overline{pz}|\geq 2\cdot |g_1|$. As $|\overline{qz}|\geq |\overline{pz}|$, the observation follows.
\end{proof}

\begin{figure}[t]
\begin{minipage}[t]{0.49\textwidth}
\begin{center}
\includegraphics[height=2.1in]{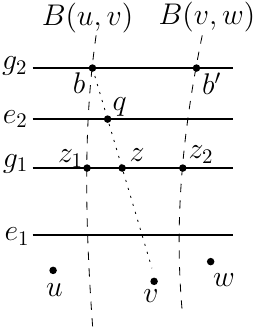}
\caption{\footnotesize Illustrating the two bisectors $B(u,v)$ and $B(v,w)$ as well as the four edges $e_1$, $g_1$, $e_2$, and $g_2$ (each of these edges may also be vertical).}
\label{fig:bisectors}
\end{center}
\end{minipage}
\hspace{0.02in}
\begin{minipage}[t]{0.49\textwidth}
\begin{center}
\includegraphics[height=1.6in]{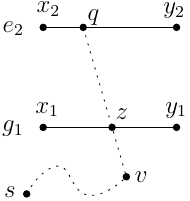}
\caption{\footnotesize Illustrating the proof of Lemma~\ref{lem:covertime}.}
\label{fig:covertime}
\end{center}
\end{minipage}
\vspace{-0.15in}
\end{figure}

The following lemma, which is a consequence of Observation~\ref{obser:length}, leads to a contradiction, for $covertime(g_1)\geq covertime(e_2)$, and thus Lemma~\ref{lem:gentriple} is proved.

\begin{lemma}\label{lem:covertime}
$covertime(g_1)<covertime(e_2)$.
\end{lemma}
\begin{proof}
Let $x_1$ and $y_1$ be the two endpoints of $g_1$, respectively. Let $x_2$ and $y_2$ be the two endpoints of $e_2$, respectively. Refer to Fig.~\ref{fig:covertime}. As $\overline{zv}\subseteq \overline{qv}$ is in the free space $\calF$, we have the following for $covertime(g_1)$:
\begin{equation*}
\begin{split}
covertime(g_1) & =\min\{d(s,x_1),d(s,y_1)\} + |g_1| \\
& \leq d(s,v) + |\overline{vz}| + |g_1|/2 + |g_1|\\
& =d(s,v) + |\overline{vz}| + 3|g_1|/2.
\end{split}
\end{equation*}

On the other hand, because $v$ claims $b$ in $W(e_2,g_2)$ (as explained in the proof of Observation~\ref{obser:intersect}), $v$ is also a generator in $W(e_2)$, and $\overline{bv}$ intersects $e_2$ at $q$, $v$ must claim $q$ in $W(e_2)$.
Hence, $\min\{d(s,x_2),d(s,y_2)\}+|e_2|\geq d(s,v)+|\overline{vq}|$ must hold, since otherwise the artificial wavelets at the endpoints of $e_2$ would have prevented $v$ from claiming $q$~\cite{ref:HershbergerAn99}.
Therefore, we have the following for $covertime(e_2)$:
\begin{equation*}
\begin{split}
covertime(e_2) & =\min\{d(s,x_2),d(s,y_2)\} + |e_2| \\
& \geq d(s,v) + |\overline{vq}|   \\
& = d(s,v) + |\overline{vz}| + |\overline{zq}| \\
& \geq d(s,v) + |\overline{vz}| + 2\cdot |g_1|.
\end{split}
\end{equation*}
The last inequality is due to Observation~\ref{obser:length}. As $|g_1|>0$, we obtain $covertime(g_1)<covertime(e_2)$.
\end{proof}

\subsection{The SPM-construction step}
\label{sublem:spm}

The above shows that the wavefront expansion step of the HS algorithm can be implemented in $O(n\log n)$ time and $O(n)$ space. As discussed in Section~\ref{sec:overview}, the algorithm computes geodesic distances from $s$ to the endpoints of all transparent edges of $\calS'$. As each obstacle vertex is incident to a transparent edge of $\calS'$, geodesic distances from $s$ to all obstacle vertices are known. In this subsection, we discuss our modification to the SPM-construction step of the HS algorithm so that the shortest path map $\spm(s)$ can be constructed in  $O(n\log n)$ time and $O(n)$ space.
As discussed before, the key is to compute the vertices of $\spm(s)$ in each cell $c$ of $\calS'$.

Consider a cell $c$ of $\calS'$. To compute the vertices of $\spm(s)$ in $c$, there are two major steps. First, partition $c$ into active and inactive regions. Second, compute the vertices of $\spm(s)$ in each active region. The original HS algorithm achieves this using the wavefronts $W(e)$ of the transparent edges $e$ on the boundary of $c$ as well as the marked generators for $c$, which are marked during the wavefront expansion step. As our new algorithm does not maintain the wavefronts $W(e)$ anymore due to the space-reset, we cannot use the same algorithm as before. We instead propose the following approach.

\paragraph{The first step.}
The first step is to compute the active regions of $c$. Recall that $c$ is partitioned into active and inactive regions by bisectors defined by an unmarked generator $u$ and a marked generator $v$. Consider such a bisector $B(u,v)$. Observe that $u$ and $v$ must be from the same wavefront $W(e)$ of an edge $e$ of $\partial c$. Indeed, assume to the contrary that $u$ is from $W(e_u)$ and $v$ is from $W(e_v)$ of two different edges $e_u$ and $e_v$ of $\partial c$. Then, because $u$ and $v$ are adjacent generators along $\partial c$, the wavelet of $u$ in $W(e_u)$ must claim an endpoint of $e_u$ and the wavelet of $v$ in $W(e_v)$ must claim an endpoint of $e_v$.
Hence, both $u$ and $v$ are marked by Rule~2(a) of the generator marking rules~\cite{ref:HershbergerAn99}. But this incurs a contradiction as $u$ is unmarked.
Therefore, $u$ and $v$ must be from the same wavefront $W(e)$ and they are actually neighboring in $W(e)$.

Based on the above observation, in addition to the marking rules of the original HS algorithm (see Section~4.1.2 of~\cite{ref:HershbergerAn99}), we add the following new rule: If a generator $v$ in a wavefront $W(e)$ is marked for a cell by the original rules during the wavefront expansion step of the HS algorithm, we also mark both neighboring generators of $v$ in $W(e)$.
We call the generators marked by the new rule {\em newly-marked} generators, while the generators marked by the original rules are called {\em originally-marked} generators. If a generator is both originally-marked and newly-marked, we consider it originally-marked. Since the total number of originally-marked generators is $O(n)$~\cite{ref:HershbergerAn99}, the total number of newly-marked generators is $O(n)$.
Also the new mark rules can be easily incorporated into the wavefront expansion step without affecting its running time asymptotically.

Now consider the bisector $B(u,v)$ discussed in the above observation. The generator $u$ must be a newly-marked generator. This suggests the following method to compute the active regions of $c$. We collect all marked generators (including both newly-marked and originally-marked) for the transparent edges of $\partial c$ and order them around the boundary of $c$. For each pair of adjacent generators $u$ and $v$ in the wavefront $W(e)$ of a transparent edge $e$ such that $u$ is newly-marked and $v$ is originally marked, we compute their bisector in $c$.
These bisectors are disjoint and partition $c$ into regions, and the regions claimed by marked generators only are active regions~\cite{ref:HershbergerAn99}. The time for computing all active regions is bounded by $O(m\log m)$, where $m$ is the number of all marked generators for $c$. Therefore, the total time of this step for all cells of $\calS'$ is bounded by $O(n\log n)$ since the total number of marked generators is $O(n)$.

\paragraph{The second step.} The second step is to compute the vertices of $\spm(s)$ in each active region $R$ of $c$. For each transparent edge $e$ of $\partial R$, which is on $\partial c$, we need to restore its wavefront $W(e)$, because it is not maintained by our new algorithm. To this end, the algorithm in the above first step determines a list $L(e)$ of originally-marked generators that claim $e$. It turns out that $L(e)$ is exactly $W(e)$, as shown below.

\begin{observation}
$L(e)=W(e)$.
\end{observation}
\begin{proof}
Since each generator of $L(e)$ is marked for $e$ and $c$, $L(e)$ is a subset of $W(e)$.
On the other hand, because $R$ is an active region, $R$ is claimed by originally-marked generators only~\cite{ref:HershbergerAn99}. Since each generator of $W(e)$ claims at least one point of $e$ and $e\subseteq R$, all generators are originally-marked. Therefore, all generators of $W(e)$ are in $L(e)$.
The observation thus follows.
\end{proof}

By the above observation, the wavefronts $W(e)$ for all transparent edges of all active regions of $c$ can be restored once all active regions of $c$ are computed in the first step. Subsequently, using the same method as the original HS algorithm, the vertices of $\spm(s)$ in $c$ can be computed in $O(m\log m)$ time, where $m$ is the total number of marked generators of $c$. As remarked in Section~\ref{sec:overview}, the space is $O(m)$. As the total number of marked generators for all cells of $\calS'$ is $O(n)$, the overall time for computing all vertices of $\spm(s)$ is $O(n\log n)$ and the space is $O(n)$. Finally, the edges of $\spm(s)$ can be computed separately and then $\spm(s)$ can be assembled by a plane sweep algorithm (see Lemma~4.13~\cite{ref:HershbergerAn99} for details); this step takes $O(n\log n)$ time and $O(n)$ space.

The following theorem summarizes our main result.

\begin{theorem}
Given a source point $s$ and a set of pairwise disjoint polygonal obstacles of $n$ vertices in the plane, the shortest path map of $s$ can be constructed in $O(n\log n)$ time and $O(n)$ space.
\end{theorem}

By building a point location data structure on $\spm(s)$ in $O(n)$ time~\cite{ref:EdelsbrunnerOp86,ref:KirkpatrickOp83}, given any query point $t$, the geodesic distance $d(s,t)$ can be computed in $O(\log n)$ time and a shortest path $\pi(s,t)$ can be produced in additional time linear in the number of edges of the path.

\section{Concluding remarks}
\label{sec:con}

In this paper, by modifying the HS algorithm~\cite{ref:HershbergerAn99} we solve the Euclidean shortest path problem among polygonal obstacles in the plane in $O(n\log n)$ time and $O(n)$ space, reducing the space complexity of the HS algorithm by a logarithmic factor and settling the longstanding open question of Hershberger and Suri~\cite{ref:HershbergerAn99}. The new algorithm is now optimal in both time and space. Our main technique is to divide the HS algorithm into phases and perform a space-reset procedure after each phase by constructing a linear-space snapshot that maintains sufficient information for the subsequent algorithm.

Like the original HS algorithm, our new algorithm can also handle some extensions of the problem. For example, if the source is not a point but a more complex object like a line segment or a disk, then as discussed in~\cite{ref:HershbergerAn99}, except for initialization and propagating the initial wavelets, the rest of the algorithm does not change. As such, our new algorithm can solve this case in $O(n\log n)$ time and $O(n)$ space. For the multiple source case where there are $m$ source points and the goal is essentially to compute the {\em geodesic Voronoi diagram} of these points, i.e., partition the free space into regions so that all points in the same region have the same nearest source point, our algorithm can solve this case in $O((n+m)\log (n+m))$ time and $O(n+m)$ space. Also, as discussed in~\cite{ref:HershbergerAn99}, the algorithm still works if each source point has an initial ``delay''. Note that the geodesic Voronoi diagram problem in simple polygons has been particularly considered; other than some earlier work~\cite{ref:AronovOn89,ref:PapadopoulouA98}, the problem has recently received increased attention~\cite{ref:LiuA20,ref:OhOp19,ref:OhVo20}. Oh~\cite{ref:OhOp19} finally proposed an algorithm of $O(n+m\log m)$ time, which is optimal.

Our technique might be able to find applications elsewhere. For example, Hershberger et al.~\cite{ref:HershbergerA13} considered the shortest path problem among curved obstacles in the plane. They gave an algorithm to construct a data structure for a source point $s$ so that each shortest path query can be answered in $O(\log n)$ time. Their algorithm runs in $O(n\log n)$ time and $O(n\log n)$ space, plus $O(n)$ calls to a {\em bisector oracle}, which is to compute the intersection of two bisectors defined by pairs of curved obstacle boundary segments. Their data structure is not a shortest path map, but consists of all wavefronts (represented by persistent binary trees) produced during the algorithm after each event; the size of the data structure is $O(n\log n)$. If we are looking for a shortest path from $s$ to a single point $t$, then it might be possible that our technique can be applied to reduce the space of their algorithm to $O(n)$ because the high-level scheme of their algorithm is the same as the wavefront expansion step of the HS algorithm for polygonal obstacles. For answering queries, however, the data structure proposed in~\cite{ref:HershbergerA13} is inherently of size $\Omega(n\log n)$ because all wavefronts need to be (implicitly) maintained. To have a data structure of size $O(n)$, one possible way is to use a shortest path map as in the polygonal obstacle case. However, this will require an additional oracle to explicitly compute the bisector of two curved obstacle boundary segments; this oracle was not used by Hershberger et al.~\cite{ref:HershbergerA13} because they wanted to minimize the number of oracles their algorithm relies on.

Another possible application of our technique is on finding shortest paths on the surface of a convex polytope in three dimensions~\cite{ref:MitchellTh87,ref:MountOn84,ref:MountSt87,ref:SharirOn86,ref:SchreiberAn08}. The current best algorithm was given by Schreiber and Sharir~\cite{ref:SchreiberAn08} and their algorithm uses $O(n\log n)$ time and $O(n\log n)$ space. Although many details are different, Schreiber and Sharir~\cite{ref:SchreiberAn08} proposed an oct-tree-like 3-dimensional axis-parallel subdivision of the space, which is similar in spirit to the planar conforming subdivision in~\cite{ref:HershbergerA13}, and used it to guide their continuous Dijkstra algorithm to propagate wavefronts, in a similar manner as the HS algorithm. Schreiber and Sharir~\cite{ref:SchreiberAn08} raised an open question whether the space of their algorithm can be reduced to $O(n)$. It would be interesting to see whether our technique can be adapted to their algorithm.



\footnotesize
 \bibliographystyle{plain}
\bibliography{reference}

\end{document}